\documentclass[AMA,STIX1COL]{WileyNJD-v2}

\usepackage{amssymb}

\usepackage{graphicx}      

\usepackage{siunitx} 
\usepackage{xcolor} 
\usepackage{amssymb,amsfonts,mathtools,tensor,amsmath} 
\usepackage{blindtext} 
\usepackage{enumerate}

\usepackage{booktabs}
\usepackage{tabu}
\usepackage{multirow}

\usepackage{pgfplots}
\usepackage{varwidth}
\usetikzlibrary{arrows}
\usepackage{tikzscale}

\newcommand{\m}[1]{\boldsymbol{#1}} 
\DeclareMathOperator*{\argmin}{arg\,min} 
\newcommand{\mat}[1]{\begin{bmatrix*}[r]#1 
\end{bmatrix*}}
\newcommand{\matc}[1]{\begin{bmatrix*}[c]#1 
\end{bmatrix*}}

\DeclareMathOperator{\rank}{rank} 
\newcommand{\norm}[1]{\left\lVert #1 \right\rVert} 
\newcommand{\abs}[1]{\left\lvert #1 \right\rvert} 

\newcommand\tdd[1]{} 
\newcommand\commd[1]{} 
\newcommand\fsamd[1]{} 
\newcommand\fmichd[1]{} 

\usepackage{algorithm}
\usepackage{algpseudocode} 
\algnewcommand{\Initialize}[1]{%
	\State \textbf{initialize} {\raggedright #1}
}
\makeatletter
\algnewcommand{\Statey}[1]{\Statex \hskip\ALG@thistlm #1}
\makeatother
\algdef{SE}[DOWHILE]{Do}{doWhile}{\algorithmicdo}[1]{\algorithmicwhile\ #1}%

\articletype{Research Article}%

\received{28 October 2019}
\revised{28 October 2019}
\accepted{28 October 2019}

\begin{document}

\title{Adaptive Dynamic Programming for Model-free Tracking of Trajectories with Time-varying Parameters\protect\thanks{This is a preprint submitted to the Int J Adapt Control Signal Process. The substantially revised version will be published in the Int J Adapt Control Signal Process (DOI: 10.1002/ACS.3106).}}

\author{Florian K\"opf}

\author{Simon Ramsteiner}

\author{Michael Flad}

\author{S\"oren Hohmann}

\authormark{K\"opf \textsc{et al.}}

\address{\orgdiv{Institute of Control Systems}, \orgname{Karlsruhe Institute of Technology (KIT)}, \orgaddress{\state{Karlsruhe}, \country{Germany}}}

%

\corres{Florian K\"opf, Institute of Control Systems, Karlsruhe Institute of Technology (KIT), Kaiserstr.~12, 76131 Karlsruhe, Germany. \email{florian.koepf@kit.edu}}


\abstract[Summary]{In order to autonomously learn to control unknown systems optimally w.r.t. an objective function, Adaptive Dynamic Programming (ADP) is well-suited to adapt controllers based on experience from interaction with the system. In recent years, many researchers focused on the tracking case, where the aim is to follow a desired trajectory. So far, ADP tracking controllers assume that the reference trajectory follows time-invariant exo-system dynamics---an assumption that does not hold for many applications. In order to overcome this limitation, we propose a new Q-function which explicitly incorporates a parametrized approximation of the reference trajectory. This allows to learn to track a general class of trajectories by means of ADP. Once our Q-function has been learned, the associated controller copes with time-varying reference trajectories without need of further training and independent of exo-system dynamics. After proposing our general model-free off-policy tracking method, we provide analysis of the important special case of linear quadratic tracking. We conclude our paper with an example which demonstrates that our new method successfully learns the optimal tracking controller and outperforms existing approaches in terms of tracking error and cost.}


\jnlcitation{\cname{%
\author{K\"opf, F.}, 
\author{S. Ramsteiner}, 
\author{M. Flad}, and 
\author{S. Hohmann}} (\cyear{2019}), 
\ctitle{Adaptive Dynamic Programming for Model-free Tracking of Trajectories with Time-varying Parameters}, \cjournal{Int J Adapt Control Signal Process}, \cvol{2020}.}

\maketitle


\section{Introduction}\label{sec:introduction}
Adaptive Dynamic Programming (ADP) which is based on Reinforcement Learning has gained extensive attention as a model-free adaptive optimal control method.\cite{Lewis.2009} In ADP, pursuing the objective to minimize a cost functional, the controller adapts its behavior on the basis of interaction with an unknown system.
The present work focuses on the ADP tracking case, where a reference trajectory is intended to be followed while the system dynamics is unknown. As the long-term cost, i.e. value, of a state changes depending on the reference trajectory, a controller that has learned to solve a regulation problem cannot be directly transferred to the tracking case.

Therefore, in literature, there are several ADP tracking approaches in discrete time \cite{Luo.2016,Kiumarsi.2014, Kopf.2018, Dierks.2009} and continuous time.\cite{Modares.2014, Zhang.2017} All of these methods assume that the reference trajectory $\m{r}_k$ can be modeled by means of a time-invariant exo-system $\m{r}_{k+1}=\m{f}_{\text{ref}}(\m{r}_k)$ (and $\dot{\m{r}}(t)=\m{f}_{\text{ref}}(\m{r}(t))$, respectively). Then, an approximated value function (or Q-function) is learned in order to rate different states (or state-action combinations) w.r.t their expected long-term cost. Based on this information, approximated optimal control laws are derived.
Whenever this reference trajectory and thus the function $\m{f}_{\text{ref}}$ changes, the learned value function and consequently the controller is not valid anymore and needs to be re-trained.
Therefore, the exo-system tracking case with time-invariant reference dynamics $\m{f}_{\text{ref}}$ is not suited for all applications.\cite{vanNieuwstadt.1997} For example in autonomous driving, process engineering and human-machine collaboration, it is often required to track flexible and time-varying trajectories. In order to account for various references, the multiple-model approach presented by Kiumarsi et al.\cite{Kiumarsi.2015} uses a self-organizing map that switches between several learned models. However, in their approach, new sub-models need to be trained for each exo-system $\m{f}_{\text{ref}}$.

Thus, our idea is to define a state-action-\textit{reference} Q-function that explicitly incorporates the course of the reference trajectory in contrast to the commonly used Q-function (see e.g. Sutton and Barto\cite{Sutton.2018}) which only depends on the current state $\m{x}_k$ and control $\m{u}_k$.
This general idea has first been proposed in our previous work,\cite{Kopf.2019} where the reference $\m{r}_k$ is given on a finite horizon and assumed to be zero thereafter. Thus, the number of weights to be learned depends on the horizon on which the reference trajectory is considered. As the reference trajectory is given for each time step, this allows high flexibility, but the sampling time and (unknown) system dynamics significantly influence the reasonable horizon length and thus the number of weights to be learned. Based on these challenges, our major idea and contribution in the present work is to approximate the reference trajectory by means of a potentially time-varying parameter set $\m{P}_{k}$ in order to compress the information about the reference compared to our previous work\cite{Kopf.2019} and incorporate this parameter into a new Q-function. In doing so, the Q-function explicitly represents the dependency of the expected long-term cost on the desired reference trajectory. Hence, the associated optimal controller is able to cope with time-varying parametrized references. We term this method \textit{Parametrized Reference ADP (PRADP)}.

Our main contributions include:
\begin{itemize}
	\item The introduction of a new reference-dependent Q-function that explicitly depends on the reference-parameter $\m{P}_k$.
	\item Function approximation of this Q-function in order to realize Temporal Difference (TD) learning (cf.~Sutton\cite{Sutton.1988}).
	\item Rigorous analysis of the form of this Q-function and its associated optimal control law in the special case of linear-quadratic (LQ) tracking.
	\item A comparison of our proposed method with algorithms assuming a time-invariant exo-system $\m{f}_{\text{ref}}$ and the ground truth optimal tracking controller.
\end{itemize}\tdd{evtl. noch erwähnen, dass Vergleich mit optimalem Regler mit vollständigem Wissen erfolgt}
In the next section, the general problem definition is given. Then, PRADP is proposed in Section~\ref{sec:our_method}. Simulation results and a discussion are given in Section~\ref{sec:results} before the paper is concluded.

\section{General Problem Definition}\label{sec:problem_definition}
Consider a discrete-time controllable system
\begin{align}\label{eq:system_general}
	\m{x}_{k+1}=\m{f}\left(\m{x}_k,\m{u}_k\right),
\end{align}
where $k\in \mathbb{N}_0$ is the discrete time step, $\m{x}_k\in\mathbb{R}^n$ the system state and $\m{u}_k\in\mathbb{R}^m$ the input. The system dynamics $\m{f}(\cdot)$ is assumed to be \textit{unknown}.

Furthermore, let the parametrized reference trajectory $\m{r}(\m{P}_k,i)\in\mathbb{R}^n$ which we intend to follow be described by 
\begin{align}\label{eq:reference}
	\m{r}(\m{P}_{k},i)=\m{P}_{k} \m{\rho}(i)=\matc{\m{p}_{k,1}^\intercal\\\m{p}_{k,2}^\intercal\\ \vdots \\ \m{p}_{k,n}^\intercal}\m{\rho}(i).
\end{align}
At any time step $k$, the reference trajectory is described by means of a parameter matrix $\m{P}_{k}\in\mathbb{R}^{n\times p}$ and given basis functions $\m{\rho}(i)\in\mathbb{R}^{p}$. Here, $i\in\mathbb{N}_0$ denotes the time step on the reference from the local perspective at time $k$, i.e. for $i=0$, the reference at time step $k$ results and $i>0$ yields a prediction of the reference for future time steps. Thus, in contrast to methods which assume that the reference follows time-invariant exo-system dynamics $\m{f}_{\text{ref}}$, the parameters $\m{P}_k$ in \eqref{eq:reference} can be time-varying, allowing much more diverse reference trajectories.

Our aim is to learn a controller which does not know the system dynamics and minimizes the cost
\begin{align}\label{eq:Jk}
	J_k = \sum_{i=0}^{\infty}\gamma^{i}c(\m{x}_{k+i},\m{u}_{k+i}, \m{r}(\m{P}_k,i)),
\end{align}
where $\gamma\in [0,1)$ is a discount factor and $c(\cdot)$ denotes a non-negative one-step cost.
We define our general problem as follows.

\begin{problem}\label{problem:p1}
	For a given parametrization of the reference by means of $\m{P}_{k}$ according to \eqref{eq:reference}, an optimal control sequence that minimizes the cost~\eqref{eq:Jk} is denoted by $\m{u}_k^*,\m{u}_{k+1}^*,\dots$ and the associated cost by $J_k^*$. The system dynamics is unknown. At each time step $k$, find $\m{u}_k^*$.
\end{problem}

\section{Parametrized Reference ADP (PRADP)}\label{sec:our_method}
In order to solve Problem~\ref{problem:p1}, we first propose a new, modified Q-function whose minimizing control represents a solution $\m{u}_k^*$ to Problem~\ref{problem:p1}. In the next step, we parametrize this Q-function by means of linear function approximation. Then, we apply Least-Squares Policy Iteration (LSPI) (cf. Lagoudakis and Parr\cite{Lagoudakis.2003}) in order to learn the unknown Q-function weights from data without requiring a system model. Finally, we discuss the structure of this new Q-function for the linear-quadratic tracking problem, where analytical insights are possible.

\subsection{Proposed Q-Function}
The relative position $i$ on the current reference trajectory that is parametrized by means of $\m{P}_{k}$ according to \eqref{eq:reference} needs to be considered when minimizing the cost $J_k$ as given in \eqref{eq:Jk}. In order to do so, one could explicitly incorporate the relative time $i$ into the Q-function that is used for ADP. This would yield a Q-function of the form $Q(\m{x}_k,\m{u}_k,\m{P}_{k},i)$. However, this would unnecessarily increase the complexity of the Q-function and hence the challenge to approximate and learn such a Q-function. Thus, we decided to implicitly incorporate the relative time $i$ on the current reference trajectory parametrized by $\m{P}_{k}$ into the reference trajectory parametrization. This yields a shifted parameter matrix $\m{P}_{k}^{(i)}$ according to the following definition.
\begin{definition}{(Shifted Parameter Matrix $\m{P}_k^{(i)}$)}\label{def:shifted_P}
	Let the matrix $\m{P}_k^{(i)}$ be defined such that
	\begin{subequations}\label{eq:def_condition_trans}
		\begin{align}
			\m{r}\left(\m{P}_k^{(i)},j\right)&=\m{r}(\m{P}_k,i+j)\\
			\Leftrightarrow \m{P}_k^{(i)}\m{\rho}(j)&=\m{P}_k\m{\rho}(i+j).\label{eq:subeqb}
		\end{align}
	\end{subequations}
	Thus, 
	\begin{align}\label{eq:P_lin}
		\m{P}_k^{(i)}=\m{P}_k\m{T}(i)
	\end{align}
	is a modified version of $\m{P}_k=\m{P}_k^{(0)}$ such that the associated reference trajectory is shifted by $i$ time steps, where $\m{T}(i)$ is a suitable matrix. Note that $\m{T}(i)$ is in general ambiguous as in the general case $p>1$ the system of equations \eqref{eq:subeqb} in order to solve for $\m{P}_k^{(i)}$ is underdetermined. Thus, $\m{T}(i)$ can be any matrix such that \eqref{eq:def_condition_trans} holds.
\end{definition}

Our proposed Q-function which explicitly incorporates the reference trajectory by means of $\m{P}_k$ is given as follows.
\begin{definition}{(Parametrized Reference Q-Function)}\label{def:Q_function}
	Let
	\begin{align}\label{eq:Q-function}
		\begin{aligned}
			Q^{\m{\pi}}\left(\m{x}_k,\m{u}_k,\m{P}_k\right)&=c\left(\m{x}_k,\m{u}_k,\m{r}(\m{P}_k,0)\right)+\!\sum_{i=1}^{\infty}\!\gamma^i c\!\left(\m{x}_{k+i},\m{\pi}\!\left(\m{x}_{k+i}, \m{P}_k^{(i)}\right)\!,r(\m{P}_k,i)\!\right) \\&=c\left(\m{x}_k,\m{u}_k,\m{r}(\m{P}_k,0)\right)+\gamma Q^{\m{\pi}}\left(\m{x}_{k+1},\m{\pi}\left(\m{x}_{k+1}, \m{P}_k^{(1)}\right)\!,\m{P}_k^{(1)}\right).
		\end{aligned}
	\end{align}
\end{definition}
Here, $\m{\pi}: \mathbb{R}^{n}\times \mathbb{R}^{n\times p} \rightarrow \mathbb{R}^m$ denotes the current control policy.

Therefore, $Q^{\m{\pi}}(\m{x}_k,\m{u}_k,\m{P}_k)$ represents the accumulated discounted cost if the system is in state $\m{x}_k$, the control $\m{u}_k$ is applied at time $k$ and the policy $\m{\pi}(\cdot)$ is followed thereafter while the reference trajectory is parametrized by $\m{P}_k$. Based on \eqref{eq:Q-function}, the optimal Q-function $Q^*(\cdot)$ is given by
\begin{align}\label{eq:Q_opt}
	\begin{aligned}
		Q^{*}\left(\m{x}_k,\m{u}_k,\m{P}_k\right)&=c\left(\m{x}_k,\m{u}_k,\m{r}(\m{P}_k,0)\right)+\min_{\m{\pi}}\gamma Q^{\m{\pi}}\!\left(\m{x}_{k+1},\m{\pi}\!\left(\m{x}_{k+1}, \m{P}_k^{(1)}\right)\!\!,\m{P}_k^{(1)}\right)	
		\\&=c\left(\m{x}_k,\m{u}_k,\m{r}(\m{P}_k,0)\right)+\gamma Q^{*}\left(\m{x}_{k+1},\m{\pi}^*\!\left(\m{x}_{k+1}, \m{P}_k^{(1)}\right)\!\!,\m{P}_k^{(1)}\right).
	\end{aligned}
\end{align}
Here, the optimal control policy is denoted by $\m{\pi}^*(\cdot)$, hence $\m{\pi}^*(\m{x}_{k+1}, \m{P}_k^{(1)}) = \m{u}^*_{k+1}$. This Q-function is useful for solving Problem~\ref{problem:p1} as can be seen from the following Lemma.
\begin{lemma}\label{lem:QoptJ}
	The control $\m{u}_k$ minimizing $Q^*\left(\m{x}_k,\m{u}_k,\m{P}_k\right)$ is a solution for $\m{u}_k^*$ minimizing $J_k$ in \eqref{eq:Jk} according to Problem~\ref{problem:p1}.
\end{lemma}
\begin{proof}
	With \eqref{eq:Q_opt}
	\begin{align}
		\begin{aligned}
			\min_{\m{u}_k}Q^*\left(\m{x}_k,\m{u}_k,\m{P}_k\right)&=c\left(\m{x}_k,\m{u}_k^*,\m{r}(\m{P}_k,0)\right)+\gamma Q^*\left(\m{x}_{k+1},\m{u}_{k+1}^*,\m{P}_k^{(1)}\right)\\&=\min_{\m{u}_k,\m{u}_{k+1},\dots}\sum_{i=0}^{\infty}\gamma^i c\left(\m{x}_i,\m{u}_i,\m{r}(\m{P}_k,i)\right)\\&=J_k^*
		\end{aligned}	
	\end{align}	
	follows, which completes the proof.
\end{proof}
Thus, if the Q-function $Q^*(\m{x}_k,\m{u}_k,\m{P}_k)$ is known, the desired optimal control $\m{u}_k$ is given by
\begin{align}\label{eq:u_aus_Q}
	\m{u}_k^*=\argmin_{\m{u}_k}Q^*(\m{x}_k,\m{u}_k,\m{P}_k).
\end{align}


Lemma~\ref{lem:QoptJ} and \eqref{eq:u_aus_Q} reveal the usefulness of $Q^*\left(\m{x}_k,\m{u}_k,\m{P}_k\right)$ for solving Problem~\ref{problem:p1}. Thus, we express this Q-function by means of linear function approximation in the following. Based on the temporal-difference (TD) error, the unknown Q-function weights can then be estimated.

\subsection{Function Approximation of the Tracking Q-Function}
As classical tabular Q-learning is unable to cope with large or even continuous state and control spaces, it is common to represent the Q-function, which is assumed to be smooth, by means of function approximation \cite{Busoniu.2010}. This leads to
\begin{align}
	Q^*\left(\m{x}_k,\m{u}_k,\m{P}_k\right)=\m{w}^\intercal\m{\phi}\left(\m{x}_k,\m{u}_k,\m{P}_k\right)+\epsilon\left(\m{x}_k,\m{u}_k,\m{P}_k\right).
\end{align}
Here, $\m{w}\in\mathbb{R}^q$ is the unknown optimal weight vector, $\m{\phi}\in\mathbb{R}^q$ a vector of activation functions and $\epsilon$ the approximation error. According to the Weierstrass higher-order approximation Theorem \cite{Hornik.1990} a single hidden layer and appropriately smooth hidden layer activation functions $\m{\phi}(\cdot)$ are capable of an arbitrarily accurate approximation of the Q-function. Furthermore, if $q\rightarrow \infty$, $\epsilon \rightarrow 0$.

As $\m{w}$ is a priori unknown, let the estimated optimal Q-function be given by 
\begin{align}\label{eq:Q_paramet}
	\hat{Q}^*\left(\m{x}_k,\m{u}_k,\m{P}_k\right)=\hat{\m{w}}^\intercal\m{\phi}\left(\m{x}_k,\m{u}_k,\m{P}_k\right).
\end{align}
In analogy to \eqref{eq:u_aus_Q}, the estimated optimal control law is defined as
\begin{align}\label{eq:argmin_u_Q}
	\hat{\m{\pi}}^*(\m{x}_k, \m{P}_k)=\argmin_{\m{u}_k}\hat{Q}^*\left(\m{x}_k,\m{u}_k,\m{P}_k\right).
\end{align}
Based on this parametrization of our new Q-function, the associated TD error \cite{Sutton.1988} is defined as follows.
\begin{definition}{(TD Error of the Tracking Q-Function)}\label{def:TD}
	The TD error which results from using the estimated Q-function $\hat{Q}^*(\cdot)$ \eqref{eq:Q_paramet} in the Bellman-like equation \eqref{eq:Q_opt} is defined as
	\begin{align}\label{eq:TD_error}
		\begin{aligned}
			\delta_k &= c\left(\m{x}_k,\m{u}_k,\m{r}(\m{P}_k,0)\right)+\gamma \hat{Q}^{*}\left(\m{x}_{k+1},\hat{\m{\pi}}^*\left(\m{x}_{k+1}, \m{P}_k^{(1)}\right)\!,\m{P}_k^{(1)}\right)-\hat{Q}^{*}\left(\m{x}_k,\m{u}_k,\m{P}_k\right)\\
			&= c\left(\m{x}_k,\m{u}_k,\m{r}(\m{P}_k,0)\right)+\gamma\hat{\m{w}}^\intercal \m{\phi}\left(\m{x}_{k+1},\hat{\m{\pi}}^*\left(\m{x}_{k+1}, \m{P}_k^{(1)}\right)\!,\m{P}_k^{(1)}\right) -\hat{\m{w}}^\intercal\m{\phi}\left(\m{x}_k,\m{u}_k,\m{P}_k\right).
		\end{aligned}
	\end{align}
\end{definition}

Our goal is to estimate $\hat{\m{w}}$ in order to minimize the squared TD error $\delta_k^2$ as the TD error quantifies the quality of the Q-function approximation. However, \eqref{eq:TD_error} is scalar while $q$ weights need to be estimated. Thus, we utilize $N\geq q$ tuples
\begin{align}\label{eq:Tk}
\mathcal{T}_k &= \left\{c_k, \hat{Q}^*_k, \hat{Q}^{*+}_k\right\}, k = 1, \dots, N,\nonumber
\intertext{where}
c_k &=c\left(\m{x}_k,\m{u}_k,\m{r}(\m{P}_k,0)\right),\nonumber\\
\hat{Q}^*_k &= \hat{\m{w}}^\intercal\m{\phi}_k = \hat{\m{w}}^\intercal\m{\phi}\left(\m{x}_k,\m{u}_k,\m{P}_k\right)\nonumber
\intertext{and}
\hat{Q}^{*+}_k &= \hat{\m{w}}^\intercal\m{\phi}_k^+ = \hat{\m{w}}^\intercal \m{\phi}\left(\m{x}_{k+1},\hat{\m{\pi}}^*\left(\m{x}_{k+1}, \m{P}_k^{(1)}\right)\!,\m{P}_k^{(1)}\right)
\end{align}
from interaction with the system in order to estimate $\hat{\m{w}}$ using Least-Squares Policy Iteration (LSPI) (cf. Lagoudakis and Parr\cite{Lagoudakis.2003}). Stacking \eqref{eq:TD_error} for the tuples $\mathcal{T}_k$, $k=1,\dots, N$, yields
\begin{align}
	\begin{aligned}
		\underbrace{\matc{\delta_1\\ \vdots \\ \delta_N}}_{\m{\delta}}&=\underbrace{\matc{c_1\\ \vdots \\ c_N}}_{\m{c}}+\underbrace{\left(\gamma\matc{\m{\phi}_1^{+\intercal}\\ \vdots\\ \m{\phi}_N^{+\intercal}}-\matc{\m{\phi}_1^{\intercal}\\ \vdots\\ \m{\phi}_N^{\intercal}}\right)}_{\m{\Phi}}\hat{\m{w}}.
	\end{aligned}
\end{align}
If the excitation condition
\begin{align}\label{eq:rank_condition}
	\rank{\m{\Phi}^\intercal\m{\Phi}}=q
\end{align}
holds, $\hat{\m{w}}$ minimizing $\m{\delta}^\intercal\m{\delta}$ exists, is unique and given by
\begin{align}\label{eq:LS}
	\hat{\m{w}}=\left(\m{\Phi}^\intercal\m{\Phi}\right)^{-1}\m{\Phi}^\intercal\m{c}
\end{align}
according to {\AA}str\"om and Wittenmark, Theorem~2.1.\cite{Astrom.1995}

\begin{note}\label{note:markov}
	Using $\m{P}_k^{(1)}=\m{P}_k\m{T}(1)$ \eqref{eq:P_lin} in the training tuple $\mathcal{T}_k$ \eqref{eq:Tk} rather than an arbitrary subsequent $\m{P}_{k+1}$ guarantees (in combination with \eqref{eq:system_general}) that the Markov property holds, which is commonly required in ADP.\cite{Lewis.2009}
\end{note}

\begin{remark}
	The procedure described above is an extension of Lagoudakis and Parr, Section~5.1\cite{Lagoudakis.2003} to the tracking case where minimizing the squared TD error is targeted. In addition, an alternative projection method  described by Lagoudakis and Parr, Section~5.2\cite{Lagoudakis.2003} which targets the approximate Q-function to be a fixed point under the Bellman operator has been implemented. Both procedures yielded indistinguishable results for our linear-quadratic simulation examples but might be different in the general case.
\end{remark}

Note that $\hat{\m{\pi}}^*(\cdot)$ in $\hat{Q}^{*+}_k$ depends on $\hat{\m{w}}$, i.e. the estimation of $\hat{Q}^{*+}_k$ depends on another estimation (of the optimal control law). This mechanism is known as \textit{bootstrapping} (cf. Sutton and Barto\cite{Sutton.2018}) in Reinforcement Learning. As a consequence, rather than estimating $\hat{\m{w}}$ once by means of the least-squares estimate \eqref{eq:LS}, a \textit{policy iteration} is performed starting with $\hat{\m{w}}^{(0)}$. This procedure is given in Algorithm~\ref{alg:LSPI}, where $e_{\hat{\m{w}}}$ is a threshold for the terminal condition.

\begin{note}
	Due to the use of a Q-function which explicitly depends on the control $\m{u}_k$, this method performs off-policy learning.\cite{Sutton.2018} Thus, during training, the behavior policy (i.e. $\m{u}_k$ which is actually applied to the system) might include exploration noise in order to satisfy the rank condition \eqref{eq:rank_condition} but due to the greedy target policy $\hat{\m{\pi}}^*$ (cf. the policy improvement step \eqref{eq:argmin_u_Q}), the Q-function associated with the optimal control law is learned.
\end{note}

\begin{algorithm}[b!]
	\caption{PRADP based on LSPI}\label{alg:LSPI}
	\begin{algorithmic}[1]
		\Initialize{$i=0, \hat{\m{w}}^{(0)}$}
		\Do
		\State policy evaluation: calculate $\hat{\m{w}}^{(i+1)}$ according to \eqref{eq:LS}, where $\hat{\m{w}}=\hat{\m{w}}^{(i+1)}$
		\State policy improvement: obtain $\hat{\m{\pi}}^{(i+1)}$ from \eqref{eq:argmin_u_Q}
		\State $i = i+1$
		\doWhile{$\norm{\hat{\m{w}}^{(i)}-\hat{\m{w}}^{(i-1)}}_2>e_{\hat{\m{w}}}$}
	\end{algorithmic}
\end{algorithm}

With $\hat{Q}^{(i)}(\cdot)=\hat{\m{w}}^{(i)\intercal}\m{\phi}(\cdot)$ and $Q^{\hat{\m{\pi}}^{(i)}}$ according to \eqref{eq:Q-function} with $\m{\pi}=\hat{\m{\pi}}^{(i)}$, the following convergence properties also hold for our tracking Q-function.
\begin{theorem}{(Convergence Properties of the Q-function, cf. Lagoudakis and Parr, Theorem~7.1 \cite{Lagoudakis.2003})}\label{thm:convergence}
	Let $\bar{\epsilon}\geq 0$ bound the errors between the approximate Q-function $\hat{Q}^{(i)}$ and true Q-function $Q^{\hat{\m{\pi}}^{(i)}}$ associated with $\hat{\m{\pi}}^{(i)}$ over all iterations, i.e.
	\begin{align}
		\norm{\hat{Q}^{(i)}-Q^{\hat{\m{\pi}}^{(i)}}}_\infty\leq\bar{\epsilon},\forall i=1,2,\dots .
	\end{align}
	Then, Algorithm~\ref{alg:LSPI} yields control laws such that
	\begin{align}
		\limsup_{i\rightarrow\infty}\norm{\hat{Q}^{(i)}-Q^*}_\infty\leq\frac{2\gamma\bar{\epsilon}}{(1-\gamma)^2}.
	\end{align}
\end{theorem}
\begin{proof}
	The proof is adapted from Bertsekas and Tsitsiklis, Proposition~6.2\cite{Bertsekas.1996}.
\end{proof}
Lagoudakis and Parr\cite{Lagoudakis.2003} point out that the appropriate choice of basis functions and the sample distribution (i.e. excitation) determine $\bar{\epsilon}$.
According to Theorem~\ref{thm:convergence}, Algorithm~\ref{alg:LSPI} converges to a neighborhood of the optimal tracking Q-function under appropriate choice of basis functions $\m{\phi}(\cdot)$ and excitation. However, for general nonlinear systems \eqref{eq:system_general} and cost functions \eqref{eq:Jk}, an appropriate choice of basis functions and the number of neurons is ``more of an art than science''\cite{Wang.2017} and still an open problem. As the focus of this paper lies on the new Q-function for tracking purposes rather than tuning of neural networks, we focus on linear systems and quadratic cost functions in the following---a setting that plays an important role in control engineering. This allows analytic insights into the structure of $Q^*\left(\m{x}_k,\m{u}_k,\m{P}_k\right)$ and thus proper choice of $\m{\phi}(\cdot)$ for function approximation in order to demonstrate the effectiveness of the proposed PRADP method.

\subsection{The LQ-Tracking Case}
\tdd{Oder ab hier bis inklusive Problem~2 nach oben in problem definition? }In the following, assume
\begin{align}\label{eq:sys_linear}
	\m{x}_{k+1}&=\m{A}\m{x}_k+\m{B}\m{u}_k,
\end{align}
and
\begin{align}
	J_k &= \sum_{i=0}^{\infty}\gamma^i\left[\left(\m{x}_{k+i}-\m{r}(\m{P}_k,i)\right)^\intercal \m{Q}\left(\m{x}_{k+i}-\m{r}(\m{P}_k,i)\right)\right.\left.+\m{u}_{k+i}^\intercal\m{R}\m{u}_{k+i}\right]\nonumber\\
	&\eqqcolon\sum_{i=0}^{\infty}\gamma^i\left[\m{e}_{k,i}^\intercal \m{Q}\m{e}_{k,i}+\m{u}_{k+i}^\intercal\m{R}\m{u}_{k+i}\right]\label{eq:J_quadratic}.
\end{align}
Here, $\m{Q}$ penalizes the deviation of the state $\m{x}_{k+i}$ from the reference $\m{r}(\m{P}_k,i)$ and $\m{R}$ penalizes the control effort. Furthermore, let the following assumptions hold.

\begin{assumption}\label{asm:LQR}
	Let $\m{Q}=\m{Q}^\intercal\succeq \m{0}$, $\m{R}=\m{R}^\intercal\succ \m{0}$, $(\m{A},\m{B})$ be controllable and $(\m{C},\m{A})$ be detectable, where $\m{C}^\intercal\m{C}=\m{Q}$.
\end{assumption}
\begin{assumption}\label{asm:Tracking}
	Let the matrix $\m{T}(i)$ which defines the shifted parameter matrix $\m{P}_k^{(i)}$ according to \eqref{eq:P_lin} be such that $\abs{\lambda_j}<1$, $\forall j=1,\dots,p$,	holds, where $\lambda_j$ are the eigenvalues of $\sqrt{\gamma}\m{T}(1)$.
\end{assumption}

\begin{note}
	Assumption~\ref{asm:LQR} is rather standard in the LQ setting in order to ensure the existence and uniqueness of a stabilizing solution to the discrete-time algebraic Riccati equation associated with the regulation problem given by \eqref{eq:sys_linear} and \eqref{eq:J_quadratic} for $\m{P}_k = \m{0}$ (cf. Ku{\v{c}}era, Theorem~8\cite{Kucera.1972}). Furthermore, it is obvious that the reference trajectory $\m{r}(\m{P}_k,i)$ must be defined such that a controller exists which yields finite cost $J_k$ in order to obtain a reasonable control problem. As will be seen in Theorem~\ref{thm:feedback}, Assumption~\ref{asm:Tracking} guarantees the existence of this solution.
\end{note}

The tracking error $\m{e}_{k,i}$ can be expressed as
\begin{align}\label{eq:error_e}
	\begin{aligned}
		\m{e}_{k,i}&=\m{x}_{k+i}-\m{r}(\m{P}_k,i)=\m{x}_{k+i}-\m{P}_k^{(i)}\m{\rho}(0)=\underbrace{\matc{\m{I}_n & \quad\matc{-\m{\rho}(0)&\cdots&\m{0}\\   \vdots&\ddots&\vdots\\ \m{0}&\cdots&-\m{\rho}(0)}^\intercal}}_{\eqqcolon\m{M}}\underbrace{\matc{\m{x}_{k+i}\\\m{p}_{k,1}^{(i)}\\ \vdots\\\m{p}_{k,n}^{(i)}}}_{\eqqcolon\m{y}_{k,i}},
	\end{aligned}
\end{align}
$i=0,1,\dots$, where $\m{I}_n$ denotes the $n\times n$ identity matrix and $\m{y}_{k,i}$ the extended state. The associated optimal controller is given in the following Theorem.
\begin{theorem}{(Optimal Tracking Control Law)}\label{thm:feedback}
	Let a reference \eqref{eq:reference} with shift matrix $\m{T}(i)$ as in Definition~\ref{def:shifted_P} be given. Then,
	\begin{enumerate}[(i)]
		\item the optimal controller which minimizes \eqref{eq:J_quadratic} subject to the system dynamics \eqref{eq:sys_linear} and the reference is linear w.r.t. $\m{y}_{k,i}$ (cf. \eqref{eq:error_e}) and can be stated as
		\begin{align}\label{eq:optimal_control_in_thm}
			\m{\pi}^*(\m{x}_{k+i}, \m{P}_k^{(i)})={\m{u}}^*_{k+i}=-\m{L}\m{y}_{k,i},
		\end{align}
		$i=0,1,\dots$.
		Here, the optimal gain $\m{L}$ is given by
		\begin{align}\label{eq:L_optimal}
			\m{L}&=(\gamma\tilde{\m{B}}^\intercal\tilde{\m{S}}\tilde{\m{B}}+\m{R})^{-1}\gamma\tilde{\m{B}}^\intercal\tilde{\m{S}}\tilde{\m{A}},
		\end{align}
		where
		
		\begin{align}
			\tilde{\m{A}}&=\matc{\m{A}&\m{0}&\cdots&\m{0}\\\m{0}&\m{T}(1)^\intercal&\cdots&\m{0}\\ \vdots&\vdots&\ddots&\vdots\\\m{0}&\m{0}&\cdots&\m{T}(1)^\intercal}\!,\quad \tilde{\m{B}}=\matc{\m{B}\\ \m{0} \\ \vdots \\ \m{0}},
		\end{align}
		$\tilde{\m{A}}\in\mathbb{R}^{n(p+1)\times n(p+1)}$, $\tilde{\m{B}}\in\mathbb{R}^{n(p+1)\times m}$, $\tilde{\m{Q}}=\m{M}^\intercal\m{Q}\m{M}$ and $\tilde{\m{S}}$ denotes the solution of the discrete-time algebraic Riccati equation
		\begin{align}
			\tilde{\m{S}}=\gamma\tilde{\m{A}}^\intercal\tilde{\m{S}}\tilde{\m{A}}-\gamma\tilde{\m{A}}^\intercal\tilde{\m{S}}\tilde{\m{B}}(\m{R}+\tilde{\m{B}}^\intercal\tilde{\m{S}}\tilde{\m{B}})^{-1}\tilde{\m{B}}^\intercal\tilde{\m{S}}\tilde{\m{A}}+\tilde{\m{Q}}\!.
		\end{align}

		\item Furthermore, under Assumptions~\ref{asm:LQR}--\ref{asm:Tracking}, the optimal controller $\m{\pi}^*(\m{x}_{k+i}, \m{P}_k^{(i)})$ exists and is unique.
	\end{enumerate}
\end{theorem}
\begin{proof}
	(i) With \eqref{eq:error_e}, the discounted cost \eqref{eq:J_quadratic} can be reformulated as
	\begin{align}\label{eq:cost_reformulation}
		\begin{aligned}
			J_k &= \sum_{i=0}^{\infty}\gamma^i\left[{\m{y}_{k,i}}^\intercal\m{M}^\intercal\m{Q}\m{M}\m{y}_{k,i}+\m{u}_{k+i}^\intercal\m{R}\m{u}_{k+i}\right].
		\end{aligned}
	\end{align}
	Furthermore, note that with \eqref{eq:sys_linear} and \eqref{eq:P_lin}
	\begin{align}\label{eq:system_reformulation}
		\begin{aligned}
			\m{y}_{k,i+1}&=\matc{\m{A}\m{x}_{k+i}+\m{B}\m{u}_{k+i}\\ \m{T}(1)^\intercal\m{p}_{k,1}^{(i)}\\ \vdots \\ \m{T}(1)^\intercal\m{p}_{k,n}^{(i)}}=\tilde{\m{A}}\m{y}_{k,i}+\tilde{\m{B}}\m{u}_{k+i}	
		\end{aligned}
	\end{align}
	holds. With $\gamma$, $\tilde{\m{A}}$, $\tilde{\m{B}}$, $\tilde{\m{Q}}$ and $\m{R}$, a standard \textit{discounted} LQ regulation problem results from \eqref{eq:cost_reformulation} for the extended state $\m{y}_{k,i}$. Considering that the discounted problem is equivalent to the undiscounted problem with $\sqrt{\gamma}\tilde{\m{A}}$, $\sqrt{\gamma}\tilde{\m{B}}$, $\tilde{\m{Q}}$ and $\m{R}$ (cf. Gaitsgory et al.\cite{Gaitsgory.2018}), the given problem can be reformulated to a standard undiscounted LQ problem. For the latter, it is well-known that the optimal controller is linear w.r.t. the state (here $\m{y}_{k,i}$) and the optimal gain is given by \eqref{eq:L_optimal} (see e.g. Lewis et al., Section~2.4\cite{Lewis.2012}), thus \eqref{eq:optimal_control_in_thm} holds and the first theorem assertion follows.
	
	(ii) For the second theorem assertion, we note that the stabilizability of $(\sqrt{\gamma}\tilde{\m{A}},\sqrt{\gamma}\tilde{\m{B}})$ directly follows from Assumptions~\ref{asm:LQR}--\ref{asm:Tracking}. In addition, $\m{Q}\succeq\m{0}$ yields $\tilde{\m{Q}}\succeq\m{0}$. As $(\m{C},\m{A})$ is detectable (Assumption~\ref{asm:LQR}), with $\tilde{\m{C}}^\intercal\tilde{\m{C}}=\tilde{\m{Q}}$, $(\tilde{\m{C}},\sqrt{\gamma}\tilde{\m{A}})$ is also detectable, because all additional states in $\tilde{\m{A}}$ compared to $\m{A}$ are stable due to Assumption~\ref{asm:Tracking}. Finally, due to $\tilde{\m{Q}}\succeq\m{0}$, $\m{R}\succ\m{0}$, $(\sqrt{\gamma}\tilde{\m{A}},\sqrt{\gamma}\tilde{\m{B}})$ stabilizable and $(\tilde{\m{C}},\sqrt{\gamma}\tilde{\m{A}})$ detectable, a unique stabilizing solution exists Ku{\v{c}}era, Theorem~8.\cite{Kucera.1972}
\end{proof}
\begin{note}\label{note:DARE}
	The proof of Theorem~\ref{thm:feedback} demonstrates that in case of \textit{known} system dynamics by means of $\m{A}$ and $\m{B}$, the optimal tracking controller $\m{L}$ can be directly calculated by solving the discrete-time algebraic Riccati equation \cite{Arnold.1984} associated with $\sqrt{\gamma}\tilde{\m{A}}$, $\sqrt{\gamma}\tilde{\m{B}}$, $\tilde{\m{Q}}$ and $\m{R}$.
\end{note}
Equation \eqref{eq:system_reformulation} demonstrates that the Markov property holds (cf. Note~\ref{note:markov}). As a consequence of Theorem~\ref{thm:feedback}, for \textit{unknown} system dynamics, this yields the following problem in the LQ PRADP case.
\begin{problem}\label{problem:LQ}
	For $i=0,1,\dots$, find the linear extended state feedback control \eqref{eq:optimal_control_in_thm}
	minimizing \eqref{eq:J_quadratic} and apply ${\m{u}}^*_{k}=-\m{L}\m{y}_{k,0}$ to the unknown system \eqref{eq:sys_linear}.
\end{problem}
Before we derive the control law $\m{L}$, we analyze the structure of $\hat{Q}^*\left(\m{x}_k,\m{u}_k,\m{P}_k\right)$ associated with Problem~\ref{problem:LQ} in the following Lemma.
\begin{lemma}{(Structure of the Tracking Q-Function)}\label{thm:structure}
	The Q-function associated with Problem~\ref{problem:LQ} has the form
	\begin{align}\label{eq:Q_quad}
		\begin{aligned}
			Q^*(\m{x}_k,\m{u}_k,\m{P}_k)&=\m{z}_k^\intercal\m{H}\m{z}_k= \matc{\m{x}_k\\\m{u}_k\\\m{p}_{k,1:n}}^\intercal \matc{\m{h}_{xx}&\m{h}_{xu}&\m{h}_{xp}\\ \m{h}_{ux}&\m{h}_{uu}&\m{h}_{up}\\ \m{h}_{px}&\m{h}_{pu}&\m{h}_{pp}} \matc{\m{x}_k\\\m{u}_k\\\m{p}_{k,1:n}},
		\end{aligned}
	\end{align}
	where $\m{z}_k = \mat{\m{x}_k^\intercal&\m{u}_k^\intercal&\m{p}_{k,1:n}^\intercal}^\intercal= \mat{\m{x}_k^\intercal&\m{u}_k^\intercal&\m{p}_{k,1}^\intercal&\dots&\m{p}_{k,n}^\intercal}^\intercal$ and $\m{H}$ is chosen such that $\m{H}=\m{H}^\intercal$.
\end{lemma}
\begin{proof}
	With \eqref{eq:Q-function} and \eqref{eq:Q_opt}
	\begin{align}\label{eq:Q_in_proof}
		\begin{aligned}
			Q^*(\m{x}_k,\m{u}_k,\m{P}_k) &=c\left(\m{x}_k,\m{u}_k,\m{r}(\m{P}_k,0)\right)+\!\sum_{i=1}^{\infty}\gamma^i c\!\left(\m{x}_{k+i},\m{\pi}^*\left(\m{x}_{k+i},\m{P}_k^{(i)}\right)\!,\m{r}(\m{P}_k,i)\!\right)
		\end{aligned}
	\end{align}
	follows. With \eqref{eq:sys_linear}, \eqref{eq:optimal_control_in_thm} and \eqref{eq:P_lin} it is obvious that the states $\m{x}_{k+i}$ and controls $\m{\pi}^*(\m{x}_{k+i},\m{P}_k^{(i)})$ are linear w.r.t. $\m{z}_k$, $\forall i=0,1,\dots$. From this linear dependency and with \eqref{eq:error_e}, linearity of $\m{e}_{k,i}$ w.r.t. $\m{z}_k$, $\forall i=0,1,\dots$ results. Due to the linear dependencies of $\m{e}_{k,i}$ and $\m{\pi}^*(\cdot)$ and the quadratic structure of $c(\cdot)$ in \eqref{eq:J_quadratic}, the Q-function in \eqref{eq:Q_in_proof} is quadratic w.r.t. $\m{z}_k$, thus \eqref{eq:Q_quad} holds.
\end{proof}

As a consequence of Lemma~\ref{thm:structure}, $Q^*$ can be exactly parametrized by means of $\hat{Q}^*$ according to \eqref{eq:Q_paramet} if $\hat{\m{w}}=\m{w}$ corresponds to the non-redundant elements of $\m{H}=\m{H}^\intercal$ (doubling elements of $\hat{\m{w}}$ associated with off-diagonal elements of $\m{H}$) and $\m{\phi} = \m{z}_k\otimes \m{z}_k$, where $\otimes$ denotes the Kronecker product. Based on Lemma~\ref{thm:structure}, the optimal control law is given as follows.
\begin{theorem}{(Optimal Tracking Control Law in Terms of $\m{H}$)}\label{prop:optimal_control}
	The unique optimal extended state feedback control minimizing $J_k$ \eqref{eq:J_quadratic} is given by
	\begin{align}\label{eq:opt_control}
		\m{u}_k^*&=\m{\pi}^*(\m{x}_k,\m{P}_k)=-\m{L}\m{y}_k^{\m{P}_k}=-\m{h}_{uu}^{-1}\matc{\m{h}_{ux}&\m{h}_{up}}\matc{\m{x}_k\\ \m{p}_{k,1:n}}.
	\end{align}
\end{theorem}
\begin{proof}
	According to Lemma~\ref{lem:QoptJ}, the desired control $\m{u}_k^*$ minimizing $Q^*(\m{x}_k,\m{u}_k,\m{P}_k)$ is also minimizing $J_k$. With \eqref{eq:Q_quad} and $\m{H}=\m{H}^\intercal$, the necessary condition
	\begin{align}
		\frac{\partial Q^*(\m{x}_k,\m{u}_k,\m{P}_k)}{\partial \m{u}_k}=2\left(\m{h}_{ux}\m{x}_k+\m{h}_{up}\m{p}_{k,1:n}+\m{h}_{uu}\m{u}_k\right)\overset{!}{=}\m{0}
	\end{align}
	yields the control $\m{u}_k^*$ given in \eqref{eq:opt_control}. Furthermore,
	\begin{align}
		\frac{\partial^2 Q^*(\m{x}_k,\m{u}_k,\m{P}_k)}{\partial \m{u}_k^2}=2\m{h}_{uu}
	\end{align}
	demonstrates that $\m{h}_{uu}\succ \m{0}$ is required in order to ensure that the control $\m{u}_k^*$ \eqref{eq:opt_control} minimizes $J_k$ \eqref{eq:J_quadratic}. This will be shown in the following. Therefore, let $Q^*_{\text{reg}}(\m{x}_k,\m{u}_k)$ be the optimal Q-function related to the regulation case, i.e. where $\m{r}(\m{P}_k,i)=\m{r}(\m{0},i)=\m{0}$. Then, it is obvious that
	\begin{align}\label{eq:Q_0}
		Q^*(\m{x}_k,\m{u}_k,\m{0})=Q^*_{\text{reg}}(\m{x}_k,\m{u}_k),
	\end{align}
	$\forall \m{x}_k\in\mathbb{R}^n, \m{u}_k\in\mathbb{R}^m$, must be true. Furthermore, for the regulation case, it is well-known that
	\begin{align}\label{eq:Q_reg}
		\begin{aligned}
			Q^*_{\text{reg}}(\m{x}_k,\m{u}_k)&=\matc{\m{x}_k\\ \m{u}_k}^\intercal\! \matc{\m{h}_{\text{reg},xx}&\m{h}_{\text{reg},xu}\\\m{h}_{\text{reg},ux}&\m{h}_{\text{reg},uu}}\!\matc{\m{x}_k\\ \m{u}_k}=\matc{\m{x}_k\\ \m{u}_k}^\intercal\! \matc{\gamma\m{A}^\intercal\m{S}\m{A}+\m{Q}&\gamma\m{A}^\intercal\m{S}\m{B}\\\gamma\m{B}^\intercal\m{S}\m{A}&\gamma\m{B}^\intercal\m{S}\m{B}+\m{R}}\!\matc{\m{x}_k\\ \m{u}_k}
		\end{aligned}
	\end{align}
	holds (see e.g. Bradtke et al.\cite{Bradtke.1994}). Here, $\m{S}$ is the solution of the discrete-time algebraic Riccati equation 
	\begin{align}
		\m{S}=\gamma\m{A}^\intercal\m{S}\m{A}-\gamma\m{A}^\intercal\m{S}\m{B}(\m{R}+\m{B}^\intercal\m{S}\m{B})^{-1}\m{B}^\intercal\m{S}\m{A}+\m{Q}.
	\end{align}
	Under Assumption~\ref{asm:LQR}, $\m{S}=\m{S}^\intercal\succeq\m{0}$ exists and is unique (Ku{\v{c}}era, Theorem~8\cite{Kucera.1972}). Thus, from \eqref{eq:Q_0} and \eqref{eq:Q_reg},
	\begin{align}
		\m{h}_{uu}&=\m{h}_{\text{reg},uu}=\gamma\m{B}^\intercal\m{S}\m{B}+\m{R}\succ\m{0}
	\end{align}
	results. This completes the proof.
\end{proof}
Thus, if $\m{H}$ (or equivalently $\m{w}$) is known, both $Q^*$ and $\m{\pi}^*$ can be calculated.

\section{Results}\label{sec:results}
In order to validate our proposed PRADP tracking method, we show simulation results in the following, where the reference trajectory is parametrized by means of cubic splines\footnote{Other approximations can be used by choosing different basis functions $\m{\rho}(i)$ (e.g. linear interpolation with $\m{\rho}(i)=\matc{iT&1}^\intercal$ or zero-order hold with $\m{\rho}(i)=1$).}. Furthermore, we compare the results with an ADP tracking method from literature which assumes that the reference can be described by a time-invariant exo-system $\m{f}_{\text{ref}}(\m{r}_k)$. Finally, we compare our learned controller that does not know the system dynamics with the ground truth controller which is calculated based on full system knowledge.

\subsection{Cubic Polynomial Reference Parametrization}
We choose $\m{r}(\m{P}_k,i)$ to be a cubic polynomial w.r.t. $i$, i.e. $\m{\rho}(i)=\matc{(iT)^3&(iT)^2&iT&1}^\intercal$, where $T$ denotes the sampling time. The associated transformation in order to obtain the shifted version $\m{P}_k^{(i)}$ of $\m{P}_k$ according to Definition~\ref{def:shifted_P} thus results from the following:
\begin{align}\label{eq:shifted_cubic}
	\begin{aligned}
		\m{r}(\m{P}_k,i+j)&=\m{P}_{k}\m{\rho}(i+j)=\m{P}_{k}\matc{((i+j)T)^3\\((i+j)T)^2\\(i+j)T\\1}=\m{P}_{k}\underbrace{\matc{1&3iT&3(iT)^2&(iT)^3\\0&1&2iT&(iT)^2\\0&0&1&iT\\0&0&0&1}}_{\m{T}(i)}\m{\rho}(j)=\m{P}_{k}^{(i)}\m{\rho}(j).
	\end{aligned}
\end{align}
\tdd{Formel überprüfen, entweder Unterklammerung in $\m{T}(i)$ ändern oder hintersten Teil weg}

In order to fully describe $\m{r}(\m{P}_k,i)$, the values of $\m{P}_k$ remain to be determined. Therefore, given sampling points of the reference trajectory every $d=25$ time steps, let $\m{P}_k$, $k=dj$, $j=0,1,\dots$, result from cubic spline interpolation. In between the sampling points, let $\m{P}_{k+i}=\m{P}_k^{(i)}$, ${i=1,2,\dots,d-1}$ (cf. Definition~\ref{def:shifted_P} and \eqref{eq:shifted_cubic}). This way, the controller is provided with $\m{P}_k$ at each time step $k$ when facing Problem~\ref{problem:LQ}.
\begin{note}
	The given procedure to generate parameters $\m{P}_k$ decouples the sampling time of the controller from the availability of sampling points given for the reference trajectory (in our example only every $d=25$ time steps). 
\end{note}

\subsection{Example System}
Consider a mass-spring-damper system with $m_{\text{sys}}=\SI{0.5}{\kilogram}$, $c_{\text{sys}}=\SI{0.1}{\newton\per\meter}$ and $d_{\text{sys}}=\SI{0.1}{\kilogram\per\second}$. Discretization of this system using Tustin approximation with $T=\SI{0.1}{\second}$ yields
\begin{align}
	\m{x}_{k+1}&=
	\begingroup
	\setlength\arraycolsep{4pt}
	\mat{0.9990 & 0.0990\\-0.0198 & 0.9792}
	\endgroup
	\m{x}_k+\mat{0.0099\\0.1979}u_k.
\end{align}
Here, $x_1$ corresponds to the position, $x_2$ to the velocity of the mass $m_{\text{sys}}$ and the control $u_k$ corresponds to a force.

We desire to track the position (i.e. $x_1$), thus we set
\begin{align}
	\m{Q}&=\begingroup
	\setlength\arraycolsep{4pt}\mat{100&0\\0&0}\endgroup \text{ and } \m{R} = 1
\end{align}
in order to strongly penalize the deviation of the first state from the parametrized reference (cf. \eqref{eq:J_quadratic}) and $\gamma = 0.9$. In this example setting, Assumptions~\ref{asm:LQR}--\ref{asm:Tracking} hold.

\subsection{Simulations}
In order to investigate the benefits of our proposed PRADP tracking controller, we compare our method with an ADP tracking controller from literature,\cite{Luo.2016,Kiumarsi.2014} which assumes that the reference trajectory is generated by a time-invariant exo-system $\m{f}_{\text{ref}}(\m{r}_k)$.
Both our method (with $e_{\hat{\m{w}}}=\SI{1E-5}{}$ in Algorithm~\ref{alg:LSPI}) and the comparison method from literature are trained on data of $500$ time steps, where Gaussian noise with zero mean and standard deviation of $1$ is applied to the system input $u_k$ for excitation. Note that none of the methods requires the system dynamics \eqref{eq:sys_linear}. Let $\m{r}_0 = \mat{0&1}^\intercal$. The reference trajectory during training is
\begin{align}\label{eq:trainin_ref}
	\begin{aligned}
		\mat{r_{k+1,1}\\r_{k+1,2}}&=\m{r}_{k+1}=\m{f}_{\text{ref}}(\m{r}_k)=\underbrace{\mat{0.9988&0.0500\\-0.0500&0.9988}}_{\m{F}_{\text{ref}}}\m{r}_k\\
	\end{aligned}
\end{align}
for the comparison method and the associated spline for our method.

The learned controllers both of our method and the comparison algorithm are tested on a reference trajectory for $x_1$ that equals the sine described by $r_{k,1}$ according to \eqref{eq:trainin_ref} for the first $250$ time steps. Then, the reference trajectory deviates from this sine as is depicted in Fig.~\ref{fig:plot_x} in gray. Here, the blue crosses mark the sampling points for spline interpolation, the black dashed line depicts $x_1$ resulting from our proposed method and the red dash-dotted line shows $x_1$ for the comparison method.
\begin{figure}[b!]
	\begin{center}	
		\includegraphics[width=0.8*\columnwidth, height=10cm]{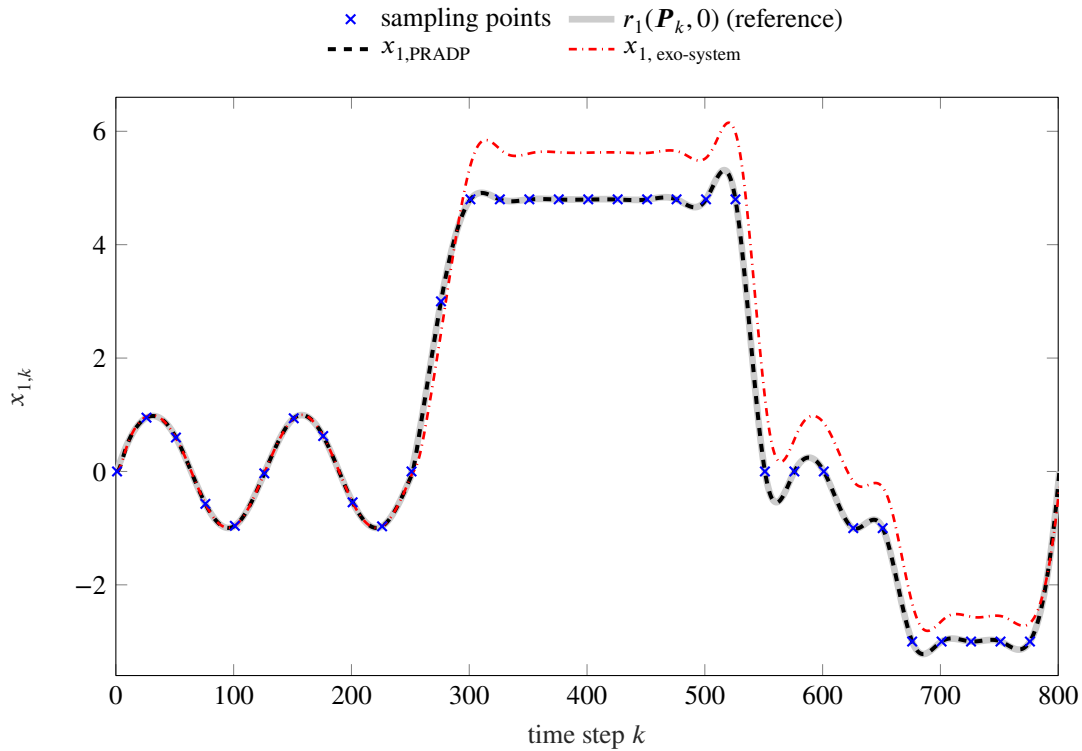}
		\caption{Tracking results of our proposed method compared with a state of the art ADP tracking controller.}\label{fig:plot_x}	
	\end{center}
\end{figure}
Furthermore, to gain insight into the tracking quality by means of the resulting cost, $c\!\left(\m{x}_k,\m{u}_k,\m{r}(\m{P}_k,0)\right)$ is depicted in Fig.~\ref{fig:compare_cost} for both methods. Note the logarithmic ordinate which is chosen in order to render the black line representing the cost associated with our method visible.

The optimal controller $\m{L}$ calculated using full system information (see Theorem~\ref{thm:feedback} and Note~\ref{note:DARE}) results in 
\begin{align}
	\m{L}=\matc{6.30&2.26&-0.31&-0.97&-2.37&-6.40&0&0&0&0}\!.
\end{align}
Comparing the learned controller $\m{L}_{\text{PRADP}}$ with the ground truth solution $\m{L}$ yields $\norm{\m{L}_{\text{PRADP}}-\m{L}}=\SI{1.08e-13}{}$. Thus, the learned controller is almost identical to the ground truth solution which demonstrates that the optimal tracking controller has successfully been learned using PRADP without knowledge of the system dynamics.

\subsection{Discussion}
As can be seen from Fig.~\ref{fig:plot_x}, our proposed method successfully tracks the parametrized reference trajectory. In contrast, the method proposed by e.g. Luo et al.\cite{Luo.2016} and Kiumarsi et al.\cite{Kiumarsi.2014} causes major deviation from the desired trajectory as soon as the reference does not follow the same exo-system which it was trained on (i.e. as soon as \eqref{eq:trainin_ref} does not hold anymore after $250$ time steps).

In addition, the cost in Fig.~\ref{fig:compare_cost} reveals that both methods yield small and similar costs as long as the reference trajectory follows $\m{F}_{\text{ref}}$. However, as soon as the reference trajectory deviates from the time-invariant exo-system description $\m{F}_{\text{ref}}$ at $k>250$, the cost of the comparison method drastically exceeds the cost associated with our proposed method. With
\begin{equation}
	\max_k{c_{\text{exo-system}}\left(\m{x}_k,\m{u}_k,\m{r}(\m{P}_k,0)\right)}\approx 270 \quad \text{ and } \quad
	\max_k{c_{\text{PRADP}}\left(\m{x}_k,\m{u}_k,\m{r}(\m{P}_k,0)\right)}\approx 2.8\nonumber,
\end{equation}
our method clearly outperforms the comparison method. PRADP does not require the assumption that the reference trajectory follows time-invariant exo-system dynamics but is nevertheless able to follow this kind of reference (see $k\leq 250$ in the simulations) as well as all other references that can be approximated by means of the time-varying parameter $\m{P}_k$. Thus, PRADP can be interpreted as a more generalized tracking approach compared to existing ADP tracking methods.
\begin{figure}[tb!]
	\begin{center}	
		\includegraphics[width=0.8*\columnwidth, height=10cm]{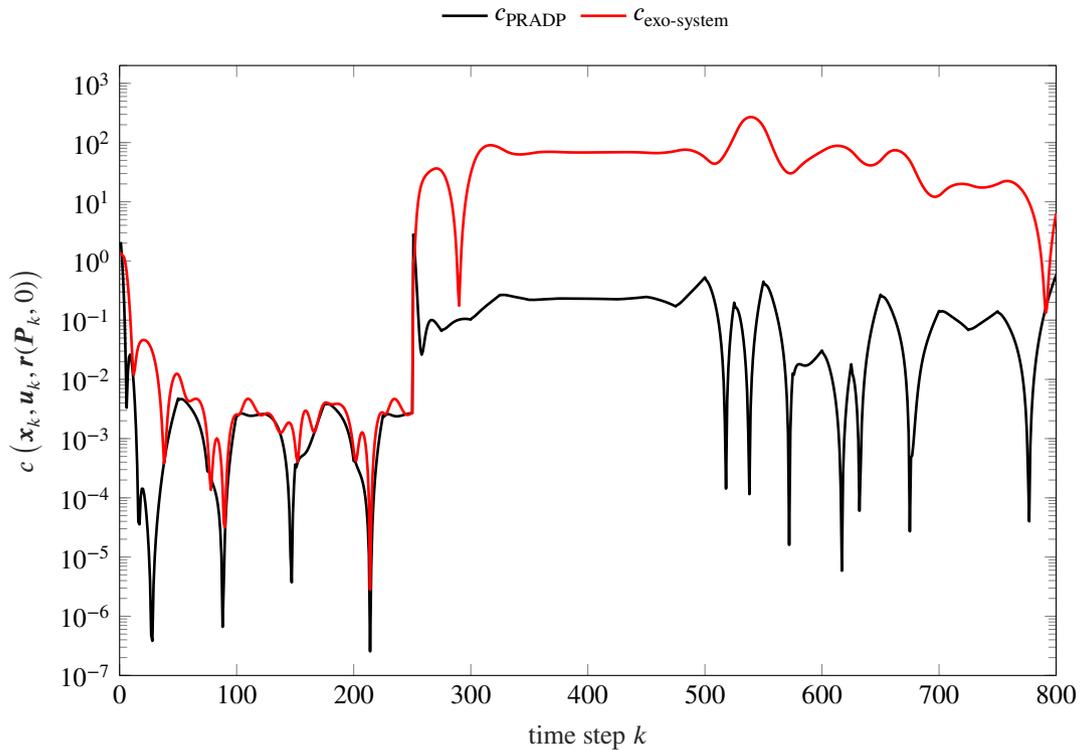}
		\caption{One-step cost $c\left(\m{x}_k,\m{u}_k,\m{r}(\m{P}_k,0)\right)$ both for our proposed method and the comparison method. Note the logarithmic ordinate.}\label{fig:compare_cost}	
	\end{center}
\end{figure}
\section{Conclusion}\label{sec:conclusion}
In this paper, we proposed a new ADP-based tracking controller termed \textit{Parametrized Reference Adaptive Dynamic Programming} (PRADP). This method implicitly incorporates the approximated reference trajectory information into the Q-function that is learned. This allows the controller to track time-varying parametrized references once the controller has been trained and does not require further adaptation or re-training compared to previous methods.
Simulation results showed that our learned controller is more flexible compared to state-of-the-art ADP tracking controllers which assume that the reference to track follows a time-invariant exo-system. Motivated by a straightforward choice of basis functions, we concentrated on the LQ tracking case in our simulations where the optimal controller has successfully been learned. However, as the mechanism of PRADP allows more general tracking problem formulations (see Section~\ref{sec:our_method}), general function approximators can be used in order to approximate $Q$ and allow for nonlinear ADP tracking controllers in the future.



%

%
%
%


%
%
%
%
%
%
%
%
%
%

\bibliography{mybib} %

\clearpage

%

\end{document}